\documentclass[12pt,draftclsnofoot,onecolumn,twoside,letterpaper]{IEEEtran}
\usepackage{enumerate}
\usepackage{multirow,array}
\usepackage{cite}
\usepackage{graphicx}
\usepackage{psfrag}
\usepackage{url}
\usepackage{amsmath}
\usepackage{amsthm}
\usepackage{array}
\usepackage{amssymb}
\usepackage{amsfonts}
\usepackage{float}

\usepackage{algorithm}
\usepackage{algpseudocode}

\usepackage{color, colortbl}
\definecolor{Gray}{gray}{0.9}

\makeatletter
\renewcommand{\boxed}[1]{\text{\fboxsep=.2em\fbox{\m@th$\displaystyle#1$}}}
\makeatother
\newcommand{\p}{\pmb}

\newtheorem{proposition}{Proposition}
\newtheorem{claim}{Claim}
\newtheorem*{conjecture*}{Conjecture}

\newtheorem{corollary}{Corollary}
\newtheorem{remark}{Remark}
\newtheorem{definition}{Definition}
\newtheorem{theorem}{Theorem}

\title{\huge An Umbrella Converse for Data Exchange: \\Applied to
Caching, Computing, and Shuffling}

  \author{Prasad Krishnan, Lakshmi Natarajan, V. Lalitha}
  
\begin{document}
\maketitle
\begin{abstract}
The problem of data exchange between multiple nodes with storage and communication capabilities models several current multi-user communication problems like Coded Caching, Data Shuffling, Coded Computing, etc. The goal in such problems is to design communication schemes which accomplish the desired data exchange between the nodes with the optimal (minimum) amount of communication load.  In this work, we present a converse to such a general data exchange problem. The expression of the converse depends only on the number of bits to be moved between different subsets of nodes, and does not assume anything further specific about the parameters in the problem. Specific problem formulations, such as those in Coded Caching, Coded Data Shuffling, Coded Distributed Computing, can be seen as instances of this generic data exchange problem. Applying our generic converse, we are able to efficiently recover known important converses in these formulations. Further, for a generic coded caching problem with heterogeneous cache sizes at the clients with or without a central server, we obtain a new general converse, which subsumes some existing results. Finally we relate a `centralized' version of our bound to the known generalized independence number bound in index coding, and discuss our bound's tightness in this context.
\end{abstract}

\let\thefootnote\relax\footnotetext{
 Dr.\ Krishnan and Dr.\ Lalitha are with the Signal Processing \& Communications Research Center, International Institute of Information Technology Hyderabad, India, email:\{prasad.krishnan,\,lalitha.v\}@iiit.ac.in.

Dr.\ Natarajan is with the Department of Electrical Engineering, Indian Institute of Technology Hyderabad, email: lakshminatarajan@iith.ac.in. 

Part of this work was presented at the IEEE Information Theory Workshop 2020 held virtually from 11–15, April 2021.

}
%%%%
\section{Introduction and Main Result} 
\label{secsystemmodel}
Consider a system of $K$ \textit{nodes}, denoted by $[K]\triangleq \{1,\hdots,K\}$, each of which have (not necessarily uniform) storage. The nodes can communicate with each other through a noiseless bus link, in which transmissions of any node is received by all others. Each node possesses a collection of data symbols (represented in bits) in its local storage, and demands another set of symbols present in other nodes. We formalize this as a \textit{data exchange problem}.
\begin{definition}
\label{dataexchangeproblemdefinition}
\textcolor{black}{A \underline{data exchange problem} on a set of $K$ nodes involving a collection $B$ of information bits is given by the following}
\begin{itemize}
    \item \textcolor{black}{a collection $\{C_i:i\in[K]\},$ where $C_i\subset B$ denotes the subset of data present in node $i$,}
    \item a collection $\{D_i:i\in[K]\}$ where \textcolor{black}{$D_i\textcolor{black}{\subset} \cup_{j\neq i}C_j\setminus C_i$} denotes the set of bits demanded by node $i$. 
\end{itemize}
% We refer to such a data exchange problem as $([K],{\mathfrak C},{\mathfrak D})$-DE problem.
\end{definition}
%%%%

The above data exchange problem models a number of cache-enabled multi-receiver communication problems studied recently in the coding theory community, including Coded Caching \cite{MaN}, Coded Distributed Computing \cite{CodedMapreduce,FundLimitsDistribCom}, Coded Data Shuffling \cite{speedingup,FundLimitsofDecentralizedDataShuffling,ElmahdyMohajerCDS_TIT}, and Coded Data Rebalancing \cite{CodedDataRebalancing}. \textcolor{black}{In \cite{Salim_CoopDataExch}, a special case of our general problem here was considered in the name of \textit{cooperative data exchange}, where the goal was to reach a state in which all nodes have all the data in the system.}

A solution to a given data exchange problem involves communication between the nodes. Each node $i$ encodes the symbols in $C_i$ into a codeword of length $l_i$ and sends it to all other nodes. The respective demanded symbols at any node is then to be decoded using the received transmissions from all the other nodes \textcolor{black}{and the node's own content.} 

Formally, a \textit{communication scheme} for the given data exchange problem consists of a set of encoding functions $\Phi\triangleq \{\phi_i:i\in[K]\}$ and decoding functions $\Psi\triangleq \{\psi_i: i\in[K]\}$, defined as follows. 
\begin{align*}
    \phi_i:&\{0,1\}^{|C_i|}\rightarrow \{0,1\}^{\textcolor{black}{l_i}}, (\text{\textcolor{black}{for some non-negative integer}}~ l_i)\\
    \psi_i:&\{0,1\}^{|C_i|}\times  \{0,1\}^{\sum\limits_{j\neq i}{\textcolor{black}{l_j}}}\rightarrow \{0,1\}^{|D_i|},
\end{align*}
%%%
such that  
\begin{align*}
\psi_i\left(C_i,\{\phi_j(C_j):j\neq i\}\right)=D_i.  
\end{align*}

The \textit{communication load} of the above scheme is defined as the total number of bits communicated, i.e., 
$$L(\Phi,\Psi)\triangleq \sum_{i\in[K]}l_i.$$ The optimal communication load is then denoted by $$L^*\triangleq \min\limits_{\Phi,\Psi}L(\Phi,\Psi).$$
%The \textit{normalized communication load} is defined as 
%$$\frac{L(\Phi,\Psi)}{|\cup_{i\in[K]} C_i|}$$
%%%%

The central result in this work is Theorem \ref{genericlowerboundTheorem} in Section \ref{dataexchangeconversesubsection}, which is a lower bound on the optimal communication load $L^*$. Using this lower bound, we recover several important converse results of cache-enabled communication problems studied in literature, including Coded Caching (Section \ref{codedcaching}), Data Shuffling (Section \ref{decentralizedcodeddatashuffling}), and Distributed Computing (Section \ref{codeddistributedcomputing}).  In each of these sections, we briefly review each setting and then apply Theorem \ref{genericlowerboundTheorem} to recover the respective converses. As a result, the proofs of these existing converses are also made simpler than what is already available in literature for the respective settings. The close relationship between these problems is quite widely known. This work gives a further formal grounding to this connection, by abstracting the common structure of these converses into a general form, which can potentially be applied to other new data exchange problems as well.

Apart from recovering existing results, more importantly we also use our data exchange lower bound to obtain \textit{new} tight converse results for some settings, while improving tightness results of some known bounds. Specifically, we present a new converse for a generic coded caching setting with multi-level cache sizes. Using this we are able to close the gap to optimality for some known special cases of this generic setting. (Section \ref{codedcachinggeneralsubsection}).
% \begin{itemize}
%     \item For the problem of Coded Data Rebalancing, introduced in \cite{CodedDataRebalancing}, we show a new converse for data rebalancing under multiple node failures.  (Section \ref{Codeddatarebalancingsection})
%     \item We present a new converse for a generic coded caching setting with multiple transmitters, receivers, and multi-level cache sizes. Using this we are able to close the gap to optimality for some known special cases of this generic setting. (Section \ref{codedcachinggeneralsubsection})
% \end{itemize}
Finally, in Section \ref{indexcodingrelationship}, we show the relationship between a `centralized' version of our data exchange lower bound and an existing bound for index coding known as the $\alpha$-bound or the generalized independence number bound \cite{ErrorCorrectionIndexCoding}. In general, we find that our bound is weaker than the $\alpha$-bound. However, for unicast index coding problems, we identify the precise conditions under which our data exchange bound is equal to the $\alpha$-bound. 
%%%%

\textit{Notations:} For positive integer \textcolor{black}{$a$}, let \textcolor{black}{$[a]\triangleq \{1,\hdots,a\}$}. For a set $S$, we denote by $S\setminus k$ the set of items in $S$ except for the item $k$, and represent the union $S\cup \{k\}$ as $S\cup k$. The binomial coefficient is denoted by $\binom{n}{k}$, which is zero if $k>n.$ \textcolor{black}{The set of all $t$-sized subsets of a set $A$ is denoted by $\binom{A}{t}.$}

\subsection{A Converse For The Data Exchange Problem}
\label{dataexchangeconversesubsection}
%%%%
In this subsection we will obtain a lower bound on the optimal communication load of the general data exchange problem defined in Section \ref{secsystemmodel}. This is the central result of this work. The predecessor to the proof technique of our data exchange lower bound is in the work of \cite{FundLimitsDistribCom}, which first presented an induction based approach for the converse of the coded distributed computing setting. Our proof uses a similar induction technique.

%We will then use this converse to various specific current problems like Coded Caching and Coded Computing to re-obtain their respective converses, as well as present a new bound on the communication load 

Given a data exchange problem and \textcolor{black}{for $P,Q\subset [K]$} such that $P\neq \phi$, let $a_P^Q$ denote the number of bits \textcolor{black}{which are stored in every node in subset of nodes $Q$ and stored in no other node, and demanded by every node in the subset $P$ and demanded by no other node}, i.e.,
\begin{align}
\label{eqndefinitionofapq}    
a_P^Q\triangleq |(\cap_{i\in P}D_i)\cap (\cap_{j\in Q}C_j)\setminus \left(\cup_{j\notin Q}C_j)\cup(\cup_{i\notin P}D_i)\right)|.
\end{align}

%Let
%\begin{align*}
%a_P^Q\triangleq |C_P^Q|.
%\end{align*}

\textcolor{black}{Note that, by defnition, $a_P^Q=0$ under the following conditions. \begin{itemize}
    \item If $P\cap Q\neq \phi$, as the bits demanded by any node are absent in the same node. 
    \item If $Q=\phi$, by Definition \ref{dataexchangeproblemdefinition}. 
    \end{itemize}} 
    %We also define $a_P^Q=0$ if $P=\phi,$ capturing the fact that if no receivers are in $P$ then there are no demanded bits at $P$ either. 
    %Note that $P$ and $Q$ must be disjoint, hence with $P=S$ gives $Q=\phi$. Similarly, if $Q=S$, then $P=\phi$ and $a_P^Q=0$.

\textcolor{black}{Theorem \ref{genericlowerboundTheorem} gives a lower bound on the optimal communication load of a given data exchange problem.  The proof of the theorem is relegated to Appendix \ref{appendix:proof_main_thm}.} The idea of the proof is as follows. If we consider only two nodes in the system, say $[K]=\{1,2\}$, then each of the $2$ nodes has to transmit whatever bits it has which are demanded by the other node, i.e. $L^*\geq a_{\{2\}}^{\{1\}}+a_{\{1\}}^{\{2\}}$. The proof of the theorem uses this as a base case and employs an induction technique to obtain a sequence of cut-set bounds leading to the final expression.

% We also see that 
% \begin{align*}
% C_i&\triangleq \bigcup_{P\subset[K]: i\in P}~\bigcup_{Q\subset[K]\setminus P} C_P^Q.
%\\
%D_i&\triangleq \bigcup_{Q\subset[K]: i\in Q}~\bigcup_{P\subset[K]\setminus Q} D_P^Q
% \end{align*} 
% \begin{remark}
% \textcolor{black}{Something needs to be said about relationship with CDC where the original bound appears} 
% \end{remark}
%%%
\begin{theorem}
\label{genericlowerboundTheorem}
$$L^*\geq \sum_{P\subset [K]}\sum_{Q\subset[K]\setminus P}\frac{|P|}{|P|+|Q|-1} a_{P}^{Q}.$$
\end{theorem}
%%%%
\textcolor{black}{Theorem \ref{genericlowerboundTheorem}, along with the observation that $a_\phi^Q=0=a^\phi_P$ gives us the following corollary, which is \textcolor{black}{a restatement} of Theorem \ref{genericlowerboundTheorem}. This will be utilized in some forthcoming results. }
%%%%
\begin{corollary}
\label{corrtomainthm}
Let $$n(p,q)\triangleq \sum\limits_{\substack{P,Q\subset[K]:\\ |P|=p,|Q|=q,P\cap Q=\phi}} a_P^Q$$ denote the total number of bits present exactly in $q$ nodes and demanded exactly by $p$ (other) nodes. Then 
\begin{align}
\label{eqn220}
    L^* \geq 
    \sum_{p=1}^{K-1}\sum_{q=1}^{K-p}\frac{p}{p+q-1}n(p,q).
\end{align}
%%%%
\end{corollary}
%%%
\begin{remark}
In \cite{QYuSli_OptimalCDC}, the authors presented an essentially identical bound (Lemma 1, \cite{QYuSli_OptimalCDC}) as Corollary \ref{corrtomainthm} in the setting of coded distributed computing. The proof given in \cite{QYuSli_OptimalCDC} for this lemma also generalizes the arguments presented in \cite{FundLimitsDistribCom} as does this work. Our present work considers a general data exchange problem and derives the lower bound in Theorem \ref{genericlowerboundTheorem} for the communication load in such a setting. \textcolor{black}{ We had derived this lower bound independently in the conference version of this paper \cite{umbrellaITW}, and only recently came to know about the bound in \cite{QYuSli_OptimalCDC}.} In subsequent sections, we show how to use this bound to recover converses for various multi-terminal communication problems considered in literature in recent years, and also obtain new converses for some settings. We also discuss, in Section \ref{indexcodingrelationship}, the looseness of Theorem \ref{genericlowerboundTheorem} by considering a centralized version of the data exchange problem and comparing our bound with the \textcolor{black}{generalized independence number bound} in index coding. These are the novel features of our present work, compared to the bound in Lemma 1 of \cite{QYuSli_OptimalCDC}. 
\end{remark}
%%%%%%
\section{Coded Caching}
\label{codedcaching}
In this section, we apply Theorem \ref{genericlowerboundTheorem} to recover the lower bound obtained in \cite{CodedCachingOptimality} for the problem of coded caching introduced in \cite{MaN}. Further using Theorem \ref{genericlowerboundTheorem}, we prove in Section \ref{codedcachinggeneralsubsection} a new converse for a generic coded caching problem under multiple cache size settings. This provides new converses for some of the existing settings in literature, and also tightens bounds in some others. In Section \ref{codedcachingmultiplefilerequests}, we recover a converse for coded caching with multiple file requests. In Section \ref{decentralizedcodedcaching}, we recover the converse for coded caching with decentralized cache placement. 

We now describe the main setting of this section. In the coded caching system introduced in \cite{MaN}, there is one server connected via a noiseless broadcast channel to $K$ clients indexed as $[K]$. The server possesses $N$ files, each of size $F$ bits, where the files are indexed by $[N]$. Each client contains local storage, or a \textit{cache}, of size $MF$ bits, for some $M\leq N$. We call this a $(K,M,N,F)$ coded caching system. The coded caching system operates in two phases: in the \textit{caching phase} which occurs during the low-traffic periods, the caches of the  clients are populated by the server with some (uncoded) bits of the file library. This is known as \textit{uncoded prefetching}. In this phase, the demands of the clients are not known. In the \textit{delivery phase} which occurs during the high-traffic periods, each client demands one file from the server, and the server makes transmissions via the broadcast channel to satisfy the client demands. The goal of coded caching is to design the caching and the delivery phase so that the worst case communication load in the delivery phase is minimized. 

For this system model, when \textcolor{black}{$\frac{MK}{N}\in {\mathbb Z}$,} the work \cite{MaN} proposed a caching and delivery scheme which achieves a communication load  (normalized by the size of the file $F$) given by $K(1-\frac{M}{N})\min\left\{\frac{1}{1+\frac{MK}{N}},\frac{N}{K}\right\}$.  In \cite{CodedCachingOptimality}, it was shown for any coded caching scheme with uncoded cache placement, communication load in the worst case (over all possible demand choices of clients) is lower bounded by $L^*\geq \frac{K(1-\frac{M}{N})F}{1+MK/N}$. Therefore it was shown that, when $K\leq N$ and \textcolor{black}{$\frac{MK}{N}\in {\mathbb Z}$}, the scheme of \cite{MaN} is optimal.

In the present section we give another proof of the lower bound for coded caching derived in \cite{CodedCachingOptimality}. The new proof that we give follows three simple steps, which are also generic to all other similar lower bounds proved in this paper for various other settings. 
\begin{itemize}
    \item \textcolor{black}{\textbf{Apply Theorem \ref{genericlowerboundTheorem}}} to the present setting obtaining a lower bound on the communication load, assuming an arbitrary choice of demands across the clients. 
    \item \textcolor{black}{\textbf{`Symmetrize' the lower bound}} by averaging over an appropriate choice of demand vectors. 
    \item \textcolor{black}{\textbf{ Refine the averaged bound}} by imposing the constraints on the total storage and using convexity of terms inside the averaged bound to obtain the final expression of the bound.
\end{itemize}
\textcolor{black}{The reader can refer to the proof of Theorem \ref{CodedCachingLowerbound} as an illustration, as we mark these three stages in the proof.} 

We now proceed with restating the lower bound from \cite{CodedCachingOptimality}. Note that these converses are typically normalized by the file size in literature, however we recall them in their non-normalized form, in order to relate them with our data exchange problem setting. 
%%%
\begin{theorem}\cite{CodedCachingOptimality}
\label{CodedCachingLowerbound}
Consider a $(K,M,N,F)$ coded caching system with $K\leq N$. The worst case communication load $L^*$ in the delivery phase (under any uncoded prefetching scheme) satisfies 
$$
L^*\geq \frac{K(1-M/N)}{1+MK/N}F.
$$
\end{theorem}
%%%
\begin{proof}[Proof based on Theorem \ref{genericlowerboundTheorem}]
Consider that after some cache placement phase, the $K$ client demands in the delivery phase are given by a demand vector $\p{d}=(d_1,\dots,d_K)$, where \textcolor{black}{$d_k\in[N]$} denotes the index of the demanded file of the client $k$.

We observe that a $(K,M,N,F)$ coded caching problem during the delivery phase satisfies Definition \ref{dataexchangeproblemdefinition} of a data exchange problem on $K+1$ nodes \textcolor{black}{indexed as $\{0,1,\hdots,K\}$, where we give the index $0$ to the server node and include this in the data exchange system.} \textcolor{black}{Before proceeding, we remark that the below proof gives a lower bound where all $K+1$ nodes in the system \textit{may} transmit, whereas in the coded caching system of \cite{MaN} \textit{only} the server \textit{can} transmit. Thus any lower bound that we obtain in this proof applies to the setting in \cite{MaN} also. } 

%Specifically we have the following in the equivalent data exchange problem  including the server node as one among the nodes in the equivalent data exchange problem. 
Clearly in the equivalent data exchange problem, node $0$ (the server) does not demand anything, but has a copy of all the bits in the entire system.
%Without loss of generality, we can assume that all transmissions are done by the source. 
Also, we are interested in bounding the worst case load over all possible demanded files at the clients. Hence we assume that all the demanded files are distinct, i.e., $d_k\neq d_{k'}$ for all $k\neq k'$. With these observations, we have by definition of $a_P^Q$ in (\ref{eqndefinitionofapq})
%%%
\begin{align}
\label{eqn221}
a_P^Q=0,~~ \text{~if~} 0 \notin Q \text{~or if~} P\notin \binom{[K]}{1},    
\end{align}
%%%
where the quantities $a_P^Q$ clearly depend on the demand vector $\p d$. 

We thus use a new set of variables: for each $k \in [K]$, $Q \subset  [K]$, and given demands $\p{d}=(d_1,\dots,d_K)$ let $c_k^Q(\p{d})$ denote the number of bits demanded by receiver node $k$ that are available only at the nodes $Q \cup \{0\}$, i.e., 
%%%
\begin{align}
\label{eqn222}
c_k^Q(\p{d})\triangleq a_{\{k\}}^{Q\cup 0}.
\end{align}
%%%%
\textcolor{black}{\textbf{Applying Theorem \ref{genericlowerboundTheorem}:}} By Theorem \ref{genericlowerboundTheorem}, we have the following lower bound for demand vector $\p d$
%%%%%
\begin{align}
\nonumber
L^* &\geq \sum_{P \subset [K] \cup \{0\}} \sum_{Q' \subset [K] \cup \{0\} \setminus P} \frac{|P|}{|P| + |Q'| - 1}a_P^{Q'}\\
\label{eqn223}
&~~~~=\sum_{k=1}^{K} \sum_{Q \subset [K] \setminus k} \frac{1}{|Q|+1} c_k^Q (\p{d}),
\end{align}
%%%%
where (\ref{eqn223}) is obtained from (\ref{eqn221}) and (\ref{eqn222}).

\textcolor{black}{\textbf{`Symmetrizing' (\ref{eqn223}) over carefully chosen demand vectors:}} We now consider the averaging of bounds of type (\ref{eqn223}) over a chosen collection of $N$ demand vectors, given by 
%%%

\textcolor{black}{\begin{align}
\label{eqn231}
{\cal D}&\triangleq \left\{\Big(j\oplus_N 0,j\oplus_N 1,\hdots,j\oplus_N (K-1)\Big):j=0,\hdots,N-1\right\}
%&=\{(1,2,\hdots,K),(2,3,\hdots,K+1),\hdots,(N,1,2\hdots,K-1)\},
\end{align}
}
where $j\oplus_{N}i\triangleq ((j+i) \! \mod N) +1.$

That is, ${\cal D}$ contains the demand vectors consisting of consecutive $K$ files, starting with each of the $N$ files as the demand of the first client. 

Averaging (\ref{eqn223}) through \textcolor{black}{the set of $N$ demand vectors in} ${\cal D}$, the lower bound we obtain is
\begin{align}
\label{eqn233}
L^*&\geq \frac{1}{N}\sum_{{\p d}\in{\cal D}}\sum_{k=1}^{K} \sum_{Q \subset [K] \setminus k} \frac{1}{|Q|+1} c_k^Q (\p{d}).
\end{align}
Let $b_n^Q$ denote the number of bits of file $n$ stored only in $Q \cup \{0\}$. Then, in the above sum, $b_n^Q=c_k^Q(\p{d})$ if and only if $d_k=n$. This happens precisely once in the collection of $N$ demand vectors in $\cal D$. Thus we have
\begin{align}
\nonumber
L^*&\geq \frac{1}{N}\sum_{k=1}^{K} \sum_{Q \subset [K] \setminus k}\sum_{{\p d}\in{\cal D}} \frac{1}{|Q|+1} c_k^Q (\p{d})\\
%\nonumber
%&\geq \frac{1}{N}\sum_{k=1}^{K} \sum_{Q \subset [K] \setminus k} \sum_{n=1}^N \frac{1}{|Q|+1} b_n^Q\\
\label{eqn224}
&= \frac{1}{N}\sum_{k=1}^{K} \sum_{Q \subset [K] \setminus k} \sum_{n=1}^N \frac{1}{|Q|+1} b_n^Q\\
\label{eqn225}
& = F\sum_{Q \subset [K]} \sum_{n=1}^N \left(\frac{K-|Q|}{|Q|+1}\right)\frac{b_n^Q}{NF} 
\end{align}
%%%
where (\ref{eqn225}) follows as for a fixed $n$ and $Q$, $k\in[K]\setminus Q$ in (\ref{eqn224}), and by multiplying and dividing by $F$. 

\textcolor{black}{\textbf{Refining the bound (\ref{eqn225}) by using the constraints of the setting:}} Now, by definition, $\sum_n \sum_{Q\subset[K]} b_n^Q=NF$, and thus $b_n^Q/NF:n\in[N],Q\subset[K]$, denotes a probability mass function. Also, $\sum_{Q\subset[K]} |Q| b_n^Q\leq KMF.$ As $(K-x)/(1+x)$ is a convex decreasing function for $x \geq 0$, using Jensen's inequality, 
we have $L^* \geq (K-x)/(1+x)$, where 
\begin{equation*}
x = \sum_n \sum_{Q\subset[K]} |Q| \frac{b_n^Q}{NF} \leq \frac{KMF}{NF} = \frac{KM}{N}.
\end{equation*}
Thus we get $L^*\geq \frac{K(1-M/N)}{1+MK/N}F,$ which completes the proof. 
\end{proof}

\textcolor{black}{\subsection{Server-based and Server-free Coded Caching with Heterogeneous Cache Sizes at Clients}
\label{codedcachinggeneralsubsection}
So far we have discussed the coded caching scenario where there is a central server containing the entire file library and the client cache sizes are homogeneous, i.e., the same at all clients. We now describe a generalization of the result in Theorem \ref{CodedCachingLowerbound} to the case of systems in which the clients have heterogeneous cache sizes, with either a centralized server present or absent. The proof of this is easily obtained from our data exchange bound in Theorem \ref{genericlowerboundTheorem}. To the best of our knowledge, a converse for this general setting is not known in the literature. Using this converse, we are able to derive new converses and tighten existing converses for various special cases of this setting, which include widely studied coded caching settings, such as device-to-device coded caching \cite{WirelessD2DNwsCodedCaching}.} 

\textcolor{black}{Consider a coded caching system with $N$ files (each of size $F$) with $K$ client nodes denoted by a set ${\cal K}_T$. We shall indicate by the value $\gamma$ the presence ($\gamma=1$) or absence ($\gamma=0$) of a centralized server in the system containing the file library. For the purpose of utilizing our data exchange bound, we assume that all the nodes in the system are capable of transmissions; thereby, any converse for this scenario is also valid for the usual coded caching scenario in which only the server (if it is present) does transmissions in the delivery phase. The set of clients ${\cal K}_T$ is partitioned into subsets ${\cal K}_{T_i}:i=1,\hdots,t$ where the nodes in subset ${\cal K}_{T_i}$ can store a fraction $\gamma_{T_i}$ of the file library. Let $|{\cal K}_{T_i}|=K_{T_i}.$ We now give our converse for this setting.}
%%%

\begin{proposition}
\label{proposition1}
For the above heterogeneous cache sizes setting, \textcolor{black}{assuming $K\leq N$,} the communication load $L^*$ for uncoded cache placement is lower bounded as follows.
\begin{align}
\label{eqn:lowerboundgeneralcodedcaching}
L^*\geq \left(\frac{K-\sum_{i=1}^tK_{T_i}\gamma_{T_i}}{\gamma+\sum_{i=1}^tK_{T_i}\gamma_{T_i}}\right)F.
\end{align}
\end{proposition}
%%%
Before giving the proof of Proposition \ref{proposition1}, we give the following remarks regarding the generality of Proposition \ref{proposition1}, the new results which arise by applying Proposition \ref{proposition1} and various results from existing literature that are subsumed or improved by it. 
%%%%
\begin{itemize}
    \item \textit{Heterogeneous Cache Sizes:} There exists a number of works discussing distinct or heterogenous client cache sizes, for instance \cite{UnequalCacheSizesPaper,CentralizedCodedCachingwithHeterogeneousCacheSizes}. However closed form expressions for the lower bound on the load seem to be lacking for such scenarios, to the best of our knowledge. Proposition \ref{proposition1} gives a lower bound for all such settings.  
    \item \textit{Device-to-Device Coded Caching:} Suppose there is no designated server in a coded caching setup, but the client nodes themselves are responsible for exchanging the information to satisfy their demands. This corresponds to the case of Device-to-Device (D2D) coded caching, first explored in \cite{WirelessD2DNwsCodedCaching}. In \cite{WirelessD2DNwsCodedCaching}, an achievable scheme was presented for the case when each (client) node has equal cache fraction $\frac{M}{N}$, and this scheme achieves a communication load of $(\frac{N}{M}-1)F$ bits. In the work \cite{OptimalityD2D}, it was shown that this communication load is optimal (for the regime of $K\leq N$) over all possible `one shot' schemes (where `one shot' refers to those schemes in which each demanded bit is decoded using the transmission only from one server), and further it was shown that the load is within a multiplicative factor of $2$ of the optimal communication load under the constraint for uncoded cache placement. We remark that the D2D setting of \cite{WirelessD2DNwsCodedCaching} corresponds to the special case of our current setting, with $\gamma=0,~t=1,~K_{T_1}=K,~$ and $\gamma_{T_1}=M/N$. By this correspondence, by applying Proposition \ref{proposition1}, we see that the load in this case is lower bounded as $\left(\frac{N}{M}-1\right)F,$ hence showing that the achievable scheme in \cite{WirelessD2DNwsCodedCaching} is exactly optimal under uncoded cache placement. The D2D scenario with heterogeneous cache sizes was explored in \cite{D2DHeterogeneousCacheSizes}, in which the optimal communication load was characterized as the solution of an optimization problem. However, no closed form expression of the load for such a scenario is mentioned. Clearly, our Proposition \ref{proposition1} gives such a lower bound, when we fix $\gamma=0$, for any number of levels $t$ of the client-side cache sizes.
\end{itemize}
 Further, the result for coded caching with a server and equal cache sizes at receivers, as in Theorem \ref{CodedCachingLowerbound}, is clearly obtained as a special case of Proposition \ref{proposition1} with $\gamma=1,~t=1,~K_{T_1}=K$ and $\gamma_{T_1}=\frac{M}{N}$.
 
%%%%
We now proceed to prove Proposition \ref{proposition1}. The proof is similar to that of Theorem \ref{CodedCachingLowerbound}.  
%%%
\begin{proof}[Proof of Proposition \ref{proposition1}]
As in the proof of Theorem \ref{CodedCachingLowerbound}, we will denote the server node as the node $0$. From Theorem \ref{genericlowerboundTheorem}, for our setting, we have
\begin{align*}
    L^*\geq \begin{cases}
    \sum_{P\subset {\cal K}_T}\sum_{Q\subset {\cal K}_T\setminus P}\frac{|P|}{|P|+|Q|-1}a_P^Q &~~~\text{if}~\gamma=0,\\
    \sum_{P\subset {\cal K}_T}\sum_{Q'\subset {\cal K}_T\cup 0\setminus P}\frac{|P|}{|P|+|Q'|-1}a_P^{Q'} &~~~\text{if}~\gamma=1.
    \end{cases}
\end{align*}
%%%
Note that if $\gamma=1$ (i.e., the server is present), then $a_P^{Q'}=0$ whenever $0\notin Q'.$ 

For a specific demand vector ${\p d}=(d_1,\hdots,d_{K})$ and for some $Q\subset {\cal K}_T$, we define the quantity $c_k^Q(\p d)$ as follows. 
\begin{align*}
    c_k^Q(\p d)=
    \begin{cases}
   a_{\{k\}}^Q=\text{Number of bits demanded by}~k~&~~~\text{if}~\gamma=0,\\
   ~~~~~~~~~~~~\text{available exclusively in}~ Q &\\
   a_{\{k\}}^{Q\cup 0}= \text{Number of bits demanded by}~k&~~~\text{if}~\gamma=1.\\
   ~~~~~~~~~~~~~\text{available exclusively in}~Q\cup 0 &
    \end{cases}
\end{align*}
%%%%

 Choosing the same special set of demand vectors ${\cal D}$ as in (\ref{eqn231}) and averaging the above lower bound over the demand vectors in $\cal D$ similar to the proof of Theorem \ref{CodedCachingLowerbound}, we obtain a bound similar to (\ref{eqn233}), 
\begin{align}
\label{eqn3001}
L^*&\geq \begin{cases}
\frac{1}{N}\sum_{{\p d}\in{\cal D}}\sum_{k\in{\cal K}_T} \sum_{Q \subset{\cal K}_T\setminus k} \frac{c_k^Q(\p d)}{|Q|}&~~~\text{if}~~\gamma=0,\\
\frac{1}{N}\sum_{{\p d}\in{\cal D}}\sum_{k\in{\cal K}_T} \sum_{Q \subset{\cal K}_T\setminus k} \frac{c_k^{Q}(\p d)}{|Q\cup 0|}&~~~\text{if}~~\gamma=1.
\end{cases}
\end{align}
Combining the two expressions in (\ref{eqn3001}), we can write a single equation which holds for $\gamma\in\{0,1\}$, 
\begin{align}
\label{eqn301}
L^*&\geq 
\frac{1}{N}\sum_{{\p d}\in{\cal D}}\sum_{k\in{\cal K}_T} \sum_{Q \subset{\cal K}_T\setminus k} \frac{c_k^Q(\p d)}{\gamma+|Q|}.
\end{align}
%%%%
We now define the term $b_n^Q$ as follows.
\begin{align}
\label{eqn3002}
    b_n^Q=
    \begin{cases}
    \text{Number of bits of file}~n~\text{available exclusively in}~ Q &~~~\text{if}~\gamma=0,\\
   \text{Number of bits of file}~n~\text{available exclusively in}~Q\cup 0 &~~~\text{if}~\gamma=1.
    \end{cases}
\end{align}
Using the above definition of $b_n^Q$ and observing that each demand vector in ${\cal D}$ has distinct components, equation (\ref{eqn301}) can be written as 
\begin{align}
L^*&\geq 
\frac{1}{N}\sum_{k\in{\cal K}_T} \sum_{Q \subset{\cal K}_T\setminus k}\sum_{{\p d}\in{\cal D}} \frac{c_k^Q(\p d)}{\gamma+|Q|} \\
&= 
\frac{1}{N}\sum_{k\in{\cal K}_T} \sum_{Q \subset{\cal K}_T\setminus k}\sum_{n=1}^N \frac{b_n^Q}{\gamma+|Q|}\\
&= 
\frac{1}{N}\sum_{Q \subset{\cal K}_T}\sum_{n=1}^N \frac{(K-|Q|)b_n^Q}{\gamma+|Q|}.
\end{align}
By the definition of $b_n^Q$ in (\ref{eqn3002}), we have $\sum_{n}\sum_{Q\subset {\cal K}_T} b_n^Q=NF$. Further $\sum_{Q\subset {\cal K}_T}|Q|b_n^Q \leq \sum\limits_{i=1}^t K_{T_i}\gamma_{T_i}NF$. Also, for $\gamma\geq 0$, the function $\frac{K-x}{\gamma+x}$ is a convex decreasing function in $x$ for $x> 0.$ Thus, using Jensen's inequality, we have $L^*\geq \frac{K-x}{\gamma+x},$ where 
\begin{align*}
x=\sum_n\sum_{Q\subset {\cal K}_T}|Q|\frac{b_n^Q}{NF}\leq \frac{\sum\limits_{i=1}^t K_{T_i}\gamma_{T_i}NF}{NF}=\sum\limits_{i=1}^t K_{T_i}\gamma_{T_i}.
\end{align*}
This completes the proof. 
\end{proof}
%%%%

\subsection{Coded Caching with Multiple File Requests}
\label{codedcachingmultiplefilerequests}
In \cite{MultipleFileRequestsUlukus}, coded caching with multiple file requests \textcolor{black}{was considered in which  each client requests any $\Delta$ files out of the $N$ files in the delivery phase. It was shown in \cite{MultipleFileRequestsUlukus} (Section V.A) that if the $\Delta K\leq N$, then the worst-case communication load can be lower bounded as }
\begin{align}
\label{eqnMultipleDemandsBound}    
L^*\geq \frac{K\Delta(1-M/N)}{1+MK/N}F.
\end{align}
The work \cite{MultipleFileRequestsUlukus} also gives an achievable scheme based on the scheme of \cite{MaN} which meets the above bound. The same lower bound can be derived using Theorem \ref{genericlowerboundTheorem} also, by following a similar procedure as that of the proof of Theorem \ref{CodedCachingLowerbound}. 

We give the proof in brief. The demand vector assumed in proof of Theorem \ref{genericlowerboundTheorem} becomes a $K\Delta$-length vector in this case, comprising of $K$ subvectors, each of length $\Delta$, capturing $\Delta$ demands for each client. The proof steps proceed as is until (\ref{eqn223}). The set $\cal D$ in (\ref{eqn231}) now contains the $K\Delta$-length vectors of consecutive file indices, cyclically constructed, starting from $(1,\hdots,K\Delta)$, i.e., 
%%%%
\begin{align}
\label{eqn234}
{\cal D}&\triangleq \left\{\Big(j\oplus_N 0,j\oplus_N 1,\hdots,j\oplus_N (K\Delta-1)\Big):j=0,\hdots,N-1\right\}.
\end{align}
\textcolor{black}{Thus, if the demand vector considered is $\boldsymbol{d}(j)\triangleq \Big(j\oplus_N 0,j\oplus_N 1,\hdots,j\oplus_N (K\Delta-1)\Big)\in {\cal D},$ then the indices of the demanded files at client $k\in[K]$, denoted by $\boldsymbol{d}_k(j),$ is given by \begin{align*}
\boldsymbol{d}_k(j)\triangleq \{j\oplus_N (k-1)\Delta,j\oplus_N (k-1)\Delta+1,\hdots,j\oplus_N (k\Delta-1)\}.
\end{align*}}

\textcolor{black}{The averaged lower bound expression similar to (\ref{eqn233}) is then obtained as 
\begin{align}
\label{eqn333}
L^*&\geq \frac{1}{N}\sum_{j=0}^{N-1}\sum_{k=1}^{K} \sum_{Q \subset [K] \setminus k} \frac{1}{|Q|+1} c_k^Q (\p{d}(j)).
\end{align}
In this expression we have $c_k^Q({\p d}(j))$ which now indicates the number of bits of $\Delta$ distinct and consecutive files indexed by $\boldsymbol{d}_k(j)$ and available exclusively at the nodes in $Q\cup 0$ ($0$ denoting the server). }

%We define for some $n\in[N]$, a set of consecutive $\Delta$ file indices starting from $n$ $${\cal N}(n)\triangleq \{n\oplus_N 0,n\oplus_N 1,\hdots,n\oplus_N (\Delta-1)\}.$$ 

\textcolor{black}{\underline{Observation:} $c_k^Q({\p d}(j)) = \sum\limits_{n'\in {\p d}_k(j)}b_{n'}^Q$ where $b_{n'}^Q$ denotes the number of bits of file $n'$ available exclusively in the nodes $Q\cup 0$, as in the proof of Theorem \ref{CodedCachingLowerbound}. }

\textcolor{black}{Now, $n'\in {\p d}_k(j)$ if and only if the file $n'$ is demanded by client $k$. By definition of ${\cal D}$, the event $n'\in {\p d}_k(j)$ happens for precisely $\Delta$ values of index $j$. From (\ref{eqn333}), applying the above observation, we have the following. 
%%%%
\begin{align}
\nonumber
L^*&\geq \frac{1}{N}\sum_{k=1}^{K} \sum_{Q \subset [K] \setminus k}\sum_{j=0}^{N-1} \frac{1}{|Q|+1} c_k^Q (\p{d}(j))\\
\nonumber
&= \frac{1}{N}\sum_{k=1}^{K} \sum_{Q \subset [K] \setminus k} \frac{1}{|Q|+1}\left(\sum_{j=0}^{N-1}\sum\limits_{n'\in {\p d}_k(j)}b_{n'}^Q\right)\\
\nonumber
&= \frac{1}{N}\sum_{k=1}^{K} \sum_{Q \subset [K] \setminus k} \sum_{n'=1}^N \frac{\Delta}{|Q|+1} b_{n'}^Q.
\end{align}
%%%
}

Comparing above equation with (\ref{eqn224}), and following similar subsequent steps as in the proof of Theorem \ref{CodedCachingLowerbound}, we have (\ref{eqnMultipleDemandsBound}).
%%%
\begin{remark}
The work \cite{SharedCachesElia} considers a coded caching setup in which $\Lambda$ caches ($\Lambda\leq K$) are shared between the $K$ clients. The special case when $\Lambda$ divides $K$ and each cache is serving exactly $\frac{K}{\Lambda}$ clients is equivalent to the scenario of the multiple file requests in \cite{MultipleFileRequestsUlukus} with $\Lambda$ clients, each demanding $\frac{K}{\Lambda}$ files. The above proof then recovers the converse for this setting, which is obtained in \cite{SharedCachesElia} (Section III.A in \cite{SharedCachesElia}). %\textcolor{black}{[Have to say something more about these situations which are not covered by our proof:]However, the work \cite{SharedCachesElia} obtains the converse for a general client-cache association profile, which is not directly derivable from our general converse in Theorem \ref{genericlowerboundTheorem}}. 
\end{remark}
%%%%
\subsection{Coded Caching with Decentralized Caching}
\label{decentralizedcodedcaching}
Theorem \ref{CodedCachingLowerbound} and the subsequent results discussed above hold for the \textit{centralized caching} framework, in which the caching phase is designed carefully in a predetermined fashion. In \cite{DecentralizedCodedCachingpaper}, the idea of decentralized placement was introduced, in which the caching phase is not coordinated centrally (this was called `decentralized coded caching' in \cite{DecentralizedCodedCachingpaper}). In this scenario, each client, independently of others, caches a fraction $\gamma=\frac{M}{N}$ of the bits in each of the $N$ files in the file library, chosen uniformly at random. For this scenario, the server (which has the file library) is responsible for the delivery phase. For the case of $K\leq N$, the authors in \cite{DecentralizedCodedCachingpaper} show a scheme which achieves a communication load ${K(1-M/N)}\frac{\left(1-(1-M/N)^K\right)}{MK/N}$. This was shown to be optimal in \cite{YMAExactRateMemoryTradeoff} and also in \cite{OntheOptimalityofTwoDecentralizedNujoom} via a connection to index coding. In the following we show that the same optimality follows easily via our Theorem \ref{genericlowerboundTheorem}. 

Assume that we have distinct demands at the $K$ clients, as in the proof of Theorem \ref{genericlowerboundTheorem}, given by the demand vector $\p d$. \textcolor{black}{We first note that by the law of large numbers, as $F$ increases,} for the decentralized cache placement, for any $k\in[K], Q\subset[K]\setminus k$, we have $$c_k^Q(\p d)=F\left(\frac{M}{N}\right)^{|Q|}\left(1-\frac{M}{N}\right)^{K-|Q|},$$ where $c_k^Q(\p d)$ is as defined in (\ref{eqn222}). Using this in (\ref{eqn223}), we get
\begingroup
\allowdisplaybreaks
\begin{align}
\nonumber
L^*&\geq \sum_{k=1}^{K} \sum_{Q \subset [K] \setminus k} \frac{1}{|Q|+1} c_k^Q (\p{d})\\
\nonumber
&\geq \sum_{k=1}^{K} \sum_{Q \subset [K] \setminus k}\frac{F}{|Q|+1}\left(\frac{M}{N}\right)^{|Q|}\left(1-\frac{M}{N}\right)^{K-|Q|}\\
\nonumber
&=F\sum_{Q \subset [K]}\left(\frac{K-|Q|}{|Q|+1}\right)\left(\frac{M}{N}\right)^{|Q|}\left(1-\frac{M}{N}\right)^{K-|Q|}\\
\nonumber
&=F\sum_{i=0}^K\binom{K}{i}\left(\frac{K-i}{i+1}\right)\left(\frac{M}{N}\right)^{i}\left(1-\frac{M}{N}\right)^{K-i}\\
\nonumber
&=F\sum_{i=0}^K\binom{K}{i+1}\left(\frac{M}{N}\right)^{i}\left(1-\frac{M}{N}\right)^{K-i}\\
\nonumber
&=F\sum_{j=1}^K\binom{K}{j}\left(\frac{M}{N}\right)^{j-1}\left(1-\frac{M}{N}\right)^{K-j+1}\\
\nonumber
&=F\frac{N}{M}\left(1-\frac{M}{N}\right)\sum_{j=1}^K\binom{K}{j}\left(\frac{M}{N}\right)^{j}\left(1-\frac{M}{N}\right)^{K-j}\\
\nonumber
&=\frac{NF}{M}\left(1-\frac{M}{N}\right)\left(1-\left(1-\frac{M}{N}\right)^K\right),
\end{align}
\endgroup
%%%%
where the last step follows as $\sum_{j=0}^K\binom{K}{j}\left(\frac{M}{N}\right)^{j}\left(1-\frac{M}{N}\right)^{K-j}=\left(\frac{M}{N}+\left(1-\frac{M}{N}\right)\right)^K=1.$ Thus, we have given an alternate proof of the optimality of the decentralized scheme in \cite{DecentralizedCodedCachingpaper}.
%%%%%
\section{Decentralized Coded Data Shuffling}
\label{decentralizedcodeddatashuffling}
In distributed machine learning systems consisting of a master and multiple worker nodes, data is distributed to the workers by the master in order to perform training of the machine learning model in a distributed manner. In general, this training process takes multiple iterations, with the workers doing some processing (like computing gradients) on their respective training data subsets. In order to ensure that the training data subset at each node is sufficiently representative of the data, and to improve the statistical performance of machine learning algorithms, shuffling of the training data between the worker nodes is implemented  after every training iteration. This is known as \textit{data shuffling}.

A coding theoretic approach to data-shuffling, which involves the master communicating coded data to the workers was presented in \cite{speedingup}. A decentralized approach to this problem was presented in \cite{FundLimitsofDecentralizedDataShuffling}, in which the workers exchange the training data amongst themselves, without involving the master node, to create a new desired partition in the next iteration. In this section, we look at this work \cite{FundLimitsofDecentralizedDataShuffling} and give a new simpler proof the lower bound on the communication load for the data shuffling. 

We first review the setting in \cite{FundLimitsofDecentralizedDataShuffling}. Consider $K$ workers in the system. The total dataset consists of $N=Kq$ data units, with $B$ bits per unit. Each node has a local cache of size $MB$ bits (such that $q\leq M\leq Kq$), out of which $q$ units are the current `active' data at any time step. The storage at node $k$ at the $(t-1)^{th}$ time instant is denoted as $Z_{k,t-1}$. We denote the active data at time $t-1$ in node $k$ as $A_{k,t-1}$. Clearly, $A_{k,t-1}\subset Z_{k,t-1}$ and $|A_{k,t-1}|=qB.$ The bits in $A_{k,t-1}:k\in[K]$ partition the entire data of $KqB$ bits.   

At time $t$, a new partition $\{A_{k,t}:k\in[K]\}$ is to be made active at the nodes $k$. The decentralized data shuffling problem is to find a delivery scheme (between workers) to shuffle the active data set $A_{k,t-1}:\forall k$ to a new partition $A_{k,t}:\forall k$, where this partition is known only at the time step $t$. Each worker computes a function of $Z_{k,t-1}$ and broadcasts it to the other workers. Using these transmissions and the locally available cache $Z_{k,t-1},$ each node is required to decode $A_{k,t}.$ As in the case of coded caching, one seeks to reduce the worst-case communication load by designing the initial storage and coded transmissions carefully. 

For the above setting, the following bound on the communication load $L^*$ was shown in \cite{FundLimitsofDecentralizedDataShuffling}. 
\begin{align}
\label{eqn401}
L^*\geq \frac{ Kq}{K-1}.\frac{K-M/q}{M/q}B.
\end{align}
The above bound was shown to be optimal for some special cases of the parameters, and order-optimal otherwise. 

We now recover this bound (\ref{eqn401}) by a simple proof using our generic lower bound Theorem \ref{genericlowerboundTheorem}. 

For $k\in[K]$ and $Q\subset[K],$ let $A_{k,t}^{Q}$ denote the subset of bits of $A_{k,t}$ available exactly at the nodes in $Q$ and not anywhere else. Note that $|A_{k,t}^{Q}|=0$ if $Q=\phi$, as each bit is necessarily present in at least one of the $K$ nodes.  

As per our bound in Theorem \ref{genericlowerboundTheorem}, we have
\begin{align*}
L^*\geq \sum_{k\in[K]}\sum_{Q\subset [K]\setminus k}\frac{|A_{k,t}^{Q}|}{|Q|}.
\end{align*}

Let the set of circular permutations of $(1,2,\hdots,K)$, apart from the identity permutation, be denoted by $\Gamma$. There are $K-1$ of them clearly. We denote an arbitrary permutation in $\Gamma$ by $\gamma$, and by $\gamma_k$ we denote the $k^{th}$ coordinate of $\gamma$. 

Now, consider the shuffle given by $\gamma\in \Gamma$, i.e., for each $k$, $A_{k,t}=A_{\gamma_k,t-1}$. For this shuffle, we have by the above equation
%%%
\begin{align}
L^*&\geq \sum_{k\in[K]}\sum_{Q\subset  [K]\setminus k}\frac{|A_{\gamma_k,t-1}^Q|}{|Q|}\\
\label{eqn991}
&=\sum_{Q\subset [K]}\sum_{k\in[K]\setminus Q}\frac{|A_{\gamma_k,t-1}^Q|}{|Q|}.
\end{align}
%%%
Now, averaging (\ref{eqn991}) over all permutations in $\Gamma$, we get
\begin{align}
    L^*&\geq \frac{1}{K-1}\sum_{\gamma\in\Gamma}\sum_{Q\subset [K]}\sum_{k\in[K]\setminus Q}\frac{|A_{\gamma_k,t-1}^Q|}{|Q|},\\
    &=
    \frac{1}{K-1}\sum_{Q\subset [K]}\sum_{k\in[K]\setminus Q}\sum_{\gamma\in\Gamma}\frac{|A_{\gamma_k,t-1}^Q|}{|Q|}.
\end{align}
As we go through all choices of $\gamma \in \Gamma,$ we see that $\gamma_k$ takes every value except $k$, i.e., $\gamma_k$ assumes each value in $[K]\setminus k$ exactly once. Also, $A_{k',t-1}^Q$ is the collection of bits of $A_{k',t-1}$ present only in $Q$. However, the bits $A_{k',t-1}$ are already presented in $k'$. Hence $|A_{k',t-1}^Q|=0$ if $k'\notin Q$. Hence, we have
%%%
\begin{align}
L^*&\geq \frac{1}{K-1}\sum_{Q\subset [K]}\sum_{k\in[K]\setminus Q}\sum_{k'\in Q}\frac{|A_{k',t-1}^Q|}{|Q|}, \\
   &= \frac{1}{K-1}\sum_{Q\subset [K]}\sum_{k'\in Q}\frac{|A_{k',t-1}^Q|(K-|Q|)}{|Q|}\\
   &=\frac{1}{K-1}\sum_{k'\in[K]}\sum_{Q\subset[K]: k'\in Q}\frac{|A_{k',t-1}^Q|(K-|Q|)}{|Q|}. 
\end{align}
Now, we have the following observations as $A_{k',t-1}^Q:k'\in[K], \{Q'\subset [K]: k'\in Q\}$ form a partition of all the $NB$ bits. 
\begin{align*}
\sum_{Q\subset [K]}\sum_{k'\in Q}|A_{k',t-1}^Q|&=NB=KqB\\
\sum_{Q\subset [K]}\sum_{k'\in Q}|A_{k',t-1}^Q||Q|&\leq KMB.
\end{align*}
%%%% 
Utilizing the above, and the fact that $\frac{K-|Q|}{|Q|}$ is a convex decreasing function in $|Q|$ (for $|Q|\geq 0$), we have 
\begin{align}
L^*&\geq \frac{KqB}{K-1}.\frac{(K-\sum_{Q\subset [K]}\sum_{k'\in Q}\frac{|A_{k',t-1}^Q|}{NB}|Q|)}{\sum_{Q\subset [K]}\sum_{k'\in Q}\frac{|A_{k',t-1}^Q|}{NB}|Q|}\\
&\geq \frac{KqB}{K-1}.\frac{K-KM/N}{KM/N}\\
&= \frac{KqB}{K-1}.\frac{K-M/q}{M/q}
\end{align}
Thus, we have recovered (\ref{eqn401}). 
%%%
\begin{remark}
We have considered the decentralized version of the coded data shuffling problem in this subsection. The centralized version of the data shuffling problem was introduced in \cite{speedingup} and its information theoretic limits were studied elaborately in \cite{ElmahdyMohajerCDS_TIT}. Our data exchange bound, when applied to the setting in \cite{ElmahdyMohajerCDS_TIT}, results in a looser converse result than that in \cite{ElmahdyMohajerCDS_TIT}. The reasons for this is explored in Section \ref{indexcodingrelationship} using the connection between our data exchange bound and the bound for index coding known in literature. 
\end{remark}
%%%%%

\section{Coded Distributed Computing}
\label{codeddistributedcomputing}
In a distributed computing setting, there are $N$ files on which the distributed computing task has to be performed by $K$ nodes. The job at hand is divided into three phases: Map, Shuffle and Reduce.\textcolor{black}{ In the shuffle phase, the nodes assigned to perform the distributed computing task exchange data. In \cite{FundLimitsDistribCom}, the authors proposed coded communication during the shuffle phase to reduce the communication load. We recollect the setting and the main converse result from \cite{FundLimitsDistribCom}, which we recover using our data exchange bound. }

\textcolor{black}{A subset $\mathcal{M}_i$ of $N$ files is assigned to $i^{\text{th}}$ node and the $i^{\text{th}}$ node computes the map functions on this subset in the map phase. We assume that the total number of map functions computed at the $K$ nodes is $rN$, where $r$ is referred to as the \textit{computation load}. In the reduce phase, a total of $W$ reduce functions are to be computed across the $K$ nodes corresponding to the $N$ files. Each node reduces the same number of functions. In this work, as in \cite{FundLimitsDistribCom}, we consider two scenarios: in the first one, each reduce function is computed exactly at one node and in the second, each reduce function is computed at $s$ nodes, where $s\geq 2$.}

\textcolor{black}{Each map function output (also referred to as intermediate output) corresponds to a particular file and a particular reduce function. For each file and each reduce function, an intermediate output of $T$ bits is obtained. To compute an assigned reduce function, each node requires the intermediate outputs of \textit{all} the files corresponding to the assigned reduce function. This means each node is missing the intermediate outputs (corresponding to the assigned reduce functions) of those files that are not assigned to it in the map phase. }

 \textcolor{black}{The intermediate outputs of each file assigned to node $i$ corresponding to all the reduce functions is available at node $i$ at the end of the map phase. These intermediate outputs at the end of the map phase are shuffled in the shuffle phase in order to deliver the missing intermediate outputs at the nodes. Let $L^*(r)$ be the total number of bits broadcasted by the $N$ nodes in the shuffle phase, minimized over all possible shuffling schemes and map function assignments with a computation load $r$. We refer to $L^*(r)$ as the \textit{minimum communication load}.}
 
 To obtain similar expressions for the communication load as in \cite{FundLimitsDistribCom}, we normalize the communication load by the total number of intermediate output bits ($=WNT$). \textcolor{black}{We consider the first scenario now, where each reduce function is computed exactly at one node.}

\begin{theorem}\cite{FundLimitsDistribCom}
The minimum communication load $L^*(r)$ incurred by a distributed computing system of $K$ nodes for a given computation load $r$, where every reduce function is computed at exactly one node and each node computes $\frac{W}{K}$ reduce functions is bounded as
\begin{equation}
\frac{L^*(r)}{WNT} \geq \frac{1}{r} \left ( 1 - \frac{r}{K} \right ).
\end{equation}
\end{theorem}
\begin{proof}
Let $\mathcal{M} = (\mathcal{M}_1, \ldots,  \mathcal{M}_K)$ denote a given map function assignment to the nodes, where $\mathcal{M}_i \subset [N]$. Let $L_{\mathcal{M}}$ denote the communication load associated with the map function assignment $\mathcal{M}$.  We will prove that
\begin{equation*}
\frac{L_{\mathcal{M}}}{WNT} \geq \sum_{j=1}^K \frac{\tilde{a}_{\mathcal{M}}^j}{N} \frac{K-j}{Kj},
\end{equation*}
where $\tilde{a}_{\mathcal{M}}^j$ denotes the number of files which are mapped at exactly $j$ nodes in $[K]$. It is easy to see that $\sum_{j=1}^K \tilde{a}_{\mathcal{M}}^j = N$ and $\sum_{j=1}^K j \tilde{a}_{\mathcal{M}}^j = rN$. 
We will apply Theorem \ref{genericlowerboundTheorem} to this setting.  \textcolor{black}{Recall that each reduce function is computed exactly at one node in our present setup. To apply Theorem \ref{genericlowerboundTheorem}, we need to ascertain the quantities $a_P^Q$ for $P,Q$ being disjoint subsets of $[K]$. To do this, we first denote by $\tilde{a}^Q$,  the number of files whose intermediate outputs are demanded by some node $k$ and available exclusively in the nodes of $Q$. Note that $\tilde{a}^Q$ is the same for any $k\in[K]\setminus Q$, since each node demands intermediate outputs of \textit{all} the files that are not mapped at the node itself.  }

\textcolor{black}{Since the number of reduce functions assigned to node $k$ is $\frac{W}{K}$ (as each reduce function is computed at exactly one node) and each intermediate output is $T$ bits, the number of intermediate output bits which are demanded by any node $k$ and available exclusively in the nodes of $Q$ are $\frac{WT}{K}\tilde{a}^Q$. Thus, for any $Q\subset [K]$, the quantities $a_P^Q$ in Theorem \ref{genericlowerboundTheorem} are given as follows. 
%%%
$$a_P^Q=\begin{cases}
 \frac{WT}{K}\tilde{a}^Q&\text{if}~P=\{k\}~\text{for some}~k\in[K]~\text{such that}~k\notin Q \\
0&\text{otherwise}
\end{cases}$$
%%%
Further note that $\sum_{Q\subset[K]:|Q|=j}\tilde{a}^Q=\tilde{a}_{\cal M}^j$ by definition of $\tilde{a}_{\cal M}^j.$ Using these and applying Theorem \ref{genericlowerboundTheorem} with the normalization factor $WNT$, we have the following inequalities.}
%%%
\begin{align*}
\frac{L_{\mathcal{M}}}{WNT} & \geq \frac{1}{WNT} \sum_{k=1}^K \sum_{Q \subset  [K] \setminus \{k\}} \frac{1}{|Q|} \tilde{a}^Q \frac{WT}{K} \\
%& = &  \frac{1}{KN} \sum_{Q \subset  [K]}  \sum_{k \in [K] \setminus Q}  \frac{1}{|Q|} \tilde{a}^Q \\
& =  \frac{1}{KN} \sum_{j =1}^K \sum_{Q \subset  [K]: |Q| = j} \sum_{k \in [K] \setminus Q}  \frac{1}{j} \tilde{a}^Q \\
%& \stackrel{(a)}{=}  \frac{1}{KN} \sum_{j =1}^K \sum_{Q \subset  [K]: |Q| = j}  \frac{K-j}{j} \tilde{a}^Q\\
&\textcolor{black}{=  \frac{1}{KN} \sum_{j =1}^K \sum_{Q \subset  [K]: |Q| = j}  \frac{K-j}{j} \tilde{a}^Q}\\
& \textcolor{black}{=  \frac{1}{KN} \sum_{j =1}^K  \frac{K-j}{j} \left ( \sum_{Q \subset  [K]: |Q| = j} \tilde{a}^Q \right ) }\\
& \textcolor{black}{= \frac{1}{KN} \sum_{j =1}^K  \frac{K-j}{j} \tilde{a}_{\mathcal{M}}^j.}
\end{align*}

Using definition of $L^*(r)$,
noting that $\frac{K-j}{j}$ is a convex decreasing function of $j$ and that $\frac{\sum_{j=1}^K{\tilde{a}_{\mathcal{M}}^j}}{N} = 1$, we have that
\begin{align*}
\frac{L^*(r)}{WNT}& \geq \frac{1}{K}   \frac{K-\sum_{j=1}^{K} \frac{j\tilde{a}_{\mathcal{M}}^j}{N}}{\sum_{j=1}^{K} \frac{j\tilde{a}_{\mathcal{M}}^j}{N}} \\
 &= \frac{1}{K} \frac{K - r}{r} = \frac{1}{r} \left ( 1 - \frac{r}{K} \right ).
\end{align*}
\end{proof}

Now, we consider the case in which each reduce function has to be computed at $s$ nodes. The total number of reduce functions is assumed to be $W$. In addition, the following assumption is made to keep the problem formulation symmetric with respect to reduce functions: every possible $s$ sized subset of $K$ nodes is assigned $\frac{W}{{K \choose s}}$ reduce functions (we assume ${K\choose s}$ divides $W$). As in the previous case, we will denote the communication load for a given map function assignment by $L_{\mathcal{M}}(s)$ and the minimum communication load under all possible map function assignments with computation load $r$ by $L^*(r,s)$. We will prove the following result which gives a lower bound on $L_{\mathcal{M}}(s)$.
\begin{proposition} \cite{FundLimitsDistribCom}
The communication load corresponding to a map function assignment $\mathcal{M}$ when each reduce function has to be computed at $s$ nodes is lower bounded as
\begin{equation}
\frac{L_{\mathcal{M}}(s)}{WNT} \geq \sum_{j=1}^K \frac{\tilde{a}_{\mathcal{M}}^j}{N} \sum_{l = \max(0, s-j)}^{\min(K-j,s)} \frac{ {K-j \choose l} {j \choose s-l}}{{K \choose s}} \frac{l}{l+j-1}.
\end{equation}
\end{proposition}
\begin{proof}
As before, we will denote by $\tilde{a}^Q$, the number of files whose map function outputs are available exclusively in the nodes of $Q$. Also, we will denote the number of intermediate output bits which are demanded exclusively by the nodes in $P$ and available exclusively in the nodes of $Q$ by $b_P^Q$. Then applying Theorem \ref{genericlowerboundTheorem}, the lower bound on the communication load  in terms of $\{b_P^Q\}$ is given by
\begin{eqnarray*}
\frac{L_{\mathcal{M}}(s)}{WNT} & \geq & \frac{1}{WNT} \sum_{P\subset [K]}\sum_{Q\subset[K]\setminus P}\frac{|P|}{|P|+|Q|-1} b_{P}^{Q}.
\end{eqnarray*}
We first interchange the above summation order and consider all sets $Q$ with $|Q| = j$ and all sets $P$ such that $|P| = l$. For $|Q|=j$, we need to count the subsets of size $s$, which form a subset of $P \cup Q$. Thus, for a fixed $j$, we can see that the range of $l$ can vary from $\max(0,s-j)$ to $\min(K-j,s)$. For a given subset $P$ of size $l$, the number of $s$ sized subsets which are contained within $P \cup Q$ and contain $P$ are ${j \choose s-l}$. Hence, the number of intermediate output bits demanded exclusively by the nodes in $P$ and available exclusively in $Q$, $b_P^Q$, is given by $b_{P}^{Q} = \tilde{a}^{Q} WT  \frac{{j \choose s-l}}{{K \choose s}}$. This is because each of the $s$-sized subset has to reduce $\tilde{a}^{Q} \frac{W}{{K \choose s}}$ functions. Using this relation, the above inequality can be rewritten as follows.
%%%
%\begin{figure*}
\begin{eqnarray}
\frac{L_{\mathcal{M}}(s)}{WNT}  & \geq & \sum_{j=1}^K \frac{1}{N} \left (\sum_{Q\subset[K]: |Q|=j} \tilde{a}^{Q} \right ) \sum_{l=\max(0,s-j)}^{\min(K-j,s)} \frac{l}{l+j-1} \left (\sum_{P \subset[K]\setminus Q: |P|=l} \frac{{j \choose s-l}}{{K \choose s}} \right ) \\
\label{eqn1001}
& = & \sum_{j=1}^K \frac{\tilde{a}_{\mathcal{M}}^j}{N} \sum_{l=\max(0,s-j)}^{\min(K-j,s)}  \frac{l}{l+j-1} \frac{{K-j \choose l}{j \choose s-l}}{{K \choose s}},
\end{eqnarray}
%\hrule 
%\end{figure*}
 where (\ref{eqn1001}) follows as $\tilde{a}_{\mathcal{M}}^j=\sum_{Q\subset[K]: |Q|=j} \tilde{a}^{Q}$. This completes the proof.
\end{proof}

The above lemma along with certain convexity arguments resulting from the constraints imposed by the computation load can be used to prove the lower bound on $L^*(r,s)$. Interested reader is referred to the converse proof of Theorem 2 in \cite{FundLimitsDistribCom} for the same.

\section{Relation to Index Coding Lower Bound}
\label{indexcodingrelationship}
We now consider the `centralized' version of the data exchange problem, where one of the nodes has a copy of all the information bits and is the lone transmitter in the system. 
% There are $K$ other nodes in the system which act as clients
We will use the index $0$ for this server node, and assume that there are $K$ other nodes in the system, with index set $[K]$, acting as clients. 
In terms of Definition~\ref{dataexchangeproblemdefinition}, this system is composed of $K+1$ nodes $\{0\} \cup [K]$, the demand $D_0$ of the server is empty, while the demands $D_i$ and the contents $C_i$ of all the clients are subsets of the contents of the server, i.e., $C_i,D_i \subset C_0$ for all $i \in [K]$.
Without loss of generality, \textcolor{black}{we assume that only the server performs all the transmissions since any coded bit that can be generated by any of the client nodes can be generated at the server itself.} Clearly, this is an index coding problem~\cite{IndexCodingTIT} with $K$ clients or receivers, the demand of the $i^{\text{th}}$ receiver is $D_i$ and its \emph{side information} is $C_i$.
When applied to this scenario, our main result Theorem~\ref{genericlowerboundTheorem} therefore provides a lower bound on the index coding communication cost.

The maximum acyclic induced subgraph (MAIS) and its generalization, which is known as the generalized independence number or the $\alpha$-bound, are well known lower bounds in index coding~\cite{IndexCodingTIT,ErrorCorrectionIndexCoding}.
In this section, we describe the relation between the $\alpha$-bound of index coding and the centralized version of Theorem~\ref{genericlowerboundTheorem}. We show that the latter is in general weaker, and identify the scenarios when these two bounds are identical.
We then use these observations to explain why Theorem~\ref{genericlowerboundTheorem} can not provide a tight lower bound for the centralized data shuffling problem~\cite{ElmahdyMohajerCDS_TIT}. 

Let us first apply Theorem~\ref{genericlowerboundTheorem} to the centralized data exchange problem. Since node $0$ contains all the information bits and its demand is empty, we have $a_{P}^{Q'} = 0$ if $0 \notin Q'$ or $0 \in P$. Using $Q = Q' \setminus \{0\}$ and defining the variable $c_P^Q = a_P^{Q \cup \{0\}} = a_P^{Q'}$, we obtain % the centralized version of Theorem~\ref{genericlowerboundTheorem}
\begin{theorem} \label{genericCentralized}
The centralized version of our main result Theorem~\ref{genericlowerboundTheorem} is
\begin{align}
L^* &\geq \sum_{P \subset [K]} \sum_{\substack{Q' \subset \{0\} \cup [K] \\ 0 \in Q', P \cap Q' = \phi}} \frac{|P|}{|P| + |Q'| - 1} a_{P}^{Q'} \nonumber \\
&= \sum_{P \subset [K]} \sum_{Q \subset [K] \setminus P} \frac{|P|}{|P| + |Q|} c_P^Q. \nonumber % \label{eq:centralized_main_theorem}
\end{align} 
\end{theorem} 
Note that it is possible to have $c_P^Q = a_P^{Q \cup \{0\}}> 0$ when $Q = \phi$.

In Section~\ref{sub:sec:index_coding:alpha} we express the generalized independence number $\alpha$ in terms of the parameters $c_P^Q$, and in Section~\ref{sub:sec:index_coding:relation} we identify the relation between our lower bound Theorem~\ref{genericCentralized} and the index coding lower bound $\alpha$.

\subsection{The generalized independence number bound} \label{sub:sec:index_coding:alpha}

Let $\gamma = (\gamma_1,\dots,\gamma_K)$ be any permutation of $[K]$, where $\gamma_i$ is the $i^{\text{th}}$ coordinate of the permutation.
Applying similar ideas as in the proof of Theorem~\ref{genericlowerboundTheorem} to the centralized scenario we obtain the following lower bound on $L^*$.
This lower bound considers the nodes in the order $\gamma_1,\dots,\gamma_K$, and for each node in this sequence it counts the number of bits that are demanded by this node which are neither demanded by and nor available as side information in any of the earlier nodes.
\begin{proposition} \label{prop:centralized:1}
For any permutation $\gamma$ of $[K]$, 
\begin{equation} \label{eq:prop:centralized:1}
L^* \geq \sum_{i=1}^{K} \, \sum_{\substack{P \subset \{\gamma_i,\dots,\gamma_K\} \\ \gamma_i \in P}} \, \sum_{Q \subset \{\gamma_{i+1},\dots,\gamma_K\}} c_P^Q.
\end{equation} 
\end{proposition}
\begin{proof}
See Appendix~\ref{app:prop:centralized:1}.
\end{proof}

A direct consequence of Proposition~\ref{prop:centralized:1} is 
\begin{equation} 
\label{eq:cent_lower_bound_gamma}
L^* \geq \max_{\gamma} \sum_{i=1}^{K} \, \sum_{\substack{P \subset \{\gamma_i,\dots,\gamma_K\} \\ \gamma_i \in P}} \, \sum_{Q \subset \{\gamma_{i+1},\dots,\gamma_K\}} c_P^Q
\end{equation} 
where the maximization is over all possible permutations on $[K]$. 

\textcolor{black}{We now recall the definition of the generalized independence number~\cite{ErrorCorrectionIndexCoding}. Denote the collection of the $c_P^Q$ information bits available exclusively at the nodes $Q \cup \{0\}$ and demanded exclusively by the nodes $P$ as $\{w_{P,m}^Q~:~m=1,\dots,c_P^Q\}$. Hence, the set of all the information bits present in the system is 
\begin{equation*}
B = \bigcup_{P \subset [K]} \bigcup_{Q \subset [K] \setminus P} \left\{w_{P,m}^Q~:~m=1,\dots,c_P^Q \right\}.
\end{equation*} 
Note that each bit is identified by a triple $(P,Q,m)$. 
%%%
\begin{definition}
A subset $\mathcal{H}$ of $B$ is a generalized independent set if and only if every subset $\mathcal{I} \subset \mathcal{H}$ satisfies the following: 
\begin{itemize}
    \item there exists a node $k \in [K]$ and an information bit in $\mathcal{I}$ such that this information bit is demanded by $k$ (and possibly some other nodes), and none of the other bits in $\mathcal{I}$ are available as side information at $k$.
\end{itemize} 
The generalized independence number $\alpha$ is the size of the largest generalized independent set.
\end{definition}
We next show that the lower bound in (\ref{eq:cent_lower_bound_gamma}) is in fact equal to the generalized independence number $\alpha$ of this index coding problem.}

\begin{theorem} \label{thm:alpha_bound}
The generalized independence number $\alpha$ satisfies
\begin{equation} \label{eq:thm:alpha_bound}
\alpha = \max_{\gamma} \sum_{i=1}^{K} \, \sum_{\substack{P \subset \{\gamma_i,\dots,\gamma_K\} \\ \gamma_i \in P}} \, \sum_{Q \subset \{\gamma_{i+1},\dots,\gamma_K\}} c_P^Q,
\end{equation}
where the maximization is over all $K!$ permutations of $[K]$.
\end{theorem}
\begin{proof}
See Appendix~\ref{app:thm:alpha_bound}.
\end{proof}

\subsection{Relation to the index coding lower bound} \label{sub:sec:index_coding:relation}

Proposition~\ref{prop:centralized:1} serves as the platform for comparing Theorem~\ref{genericCentralized} and the $\alpha$-bound.
While $\alpha$ equals the \emph{maximum} value of the bound in Proposition~\ref{prop:centralized:1} over all permutations on $[K]$, our bound in Theorem~\ref{genericCentralized} equals the \emph{average} value of the lower bound given in Proposition~\ref{prop:centralized:1} over all permutations on $[K]$. We will show this relation between Theorem~\ref{genericCentralized} and Proposition~\ref{prop:centralized:1} now.

Taking the average of the right hand side of~\eqref{eq:prop:centralized:1} with respect to all $\gamma$, we obtain
\begin{equation*}
\frac{1}{K!} \sum_{\gamma}  ~ \sum_{i=1}^{K} \sum_{\substack{P \subset \{\gamma_i,\dots,\gamma_K\} \\ \gamma_i \in P}} \sum_{Q \subset \{\gamma_{i+1},\dots,\gamma_K\}} c_P^Q.
\end{equation*}
For each choice of $P, Q \subset [K]$ with $P \cap Q = \phi$, we now count the number of times $c_P^Q$ appears in this sum. For a given $\gamma$, the inner summations include the term $c_P^Q$ if and only if the following holds
% The above sum `counts' $a_P^Q$ for a permutation $\pi$ if and only if the following is true
\begin{equation*}
\gamma_i \in P, \text{ where } i = \min\{ j \in [K]~:~\gamma_j \in P \cup Q\},
\end{equation*}
i.e., if we consider the elements $\gamma_1,\dots,\gamma_K$ in that order, the first element from $P \cup Q$ to be observed in this sequence belongs to $P$. 
Thus, for a given pair $P,Q$ the probability that a permutation $\gamma$ chosen uniformly at random includes the term $c_P^Q$ in the inner summation is $|P|/(|P| + |Q|)$.
Hence, the average of the lower bound in Proposition~\ref{prop:centralized:1} over all possible $\gamma$ is
\begin{equation*}
\sum_{P \subset [K]} \sum_{\substack{Q \subset [K] \\ P \cap Q = \phi}} \frac{|P|}{|P| + |Q|} c_P^Q,
\end{equation*}
which is exactly the bound in Theorem~\ref{genericCentralized}.

Since the bound in Theorem~\ref{genericCentralized} is obtained by averaging over all $\gamma$, instead of maximizing over all $\gamma$, we conclude that this is in general weaker than the $\alpha$-bound of index coding. The two bounds are equal if and only if the bound in Proposition~\ref{prop:centralized:1} has the same value for every permutation $\gamma$. 

Although weaker in general, we note that the bound of Theorem~\ref{genericCentralized} is easier to use than the $\alpha$-bound. 
As demonstrated by~\eqref{eqn220}, in order to use Theorem~\ref{genericCentralized} we only need to know, for each information bit, the number of nodes that contain this bit and the number of nodes that demand this bit. In comparison, this information is insufficient to evaluate the $\alpha$-bound, which also requires the identities of these nodes. 

% For example, for unicast problems, if the side information graph is a disjoint union of cliques then the two bounds are equal.

% {\bf Question:} a necessary and sufficient condition for the two bounds to be equal?

\subsection{On the tightness of Theorem~\ref{genericCentralized}}

We now consider the class of \emph{unicast} problems, i.e., problems where each bit is demanded by exactly one of the nodes. For this class of problems we characterize when Theorem~\ref{genericCentralized} yields a tight bound.

\begin{theorem} \label{thm:tightness_centralized}
For unicast problems the bound in Theorem~\ref{genericCentralized} equals $L^*$ if and only if every $S \subset [K]$ with $|S| \geq 2$ satisfies the following, $c_{\{k\}}^{S \setminus k} = c_{\{k'\}}^{S \setminus k'}$ for every $k,k' \in S$.
\end{theorem}
\begin{proof}
See Appendix~\ref{app:thm:tightness_centralized}. When the lower bound of Theorem~\ref{genericCentralized} is tight, the clique-covering based index coding scheme (see \cite{BirkandKolISCOD,IndexCodingTIT}) yields the optimal communication cost.
\end{proof}

Our main result in Theorem~\ref{genericlowerboundTheorem}, or equivalently, Theorem~\ref{genericCentralized}, does not provide a tight lower bound for centralized data shuffling problem~\cite{ElmahdyMohajerCDS_TIT}, because this problem involves scenarios that do not satisfy the tightness condition of Theorem~\ref{thm:tightness_centralized}. 
% We consider an example setting of the centralized data shuffling problem~\cite{ElmahdyMohajerCDS_TIT}, and use Theorem~\ref{thm:tightness_centralized} to show that our bound in Theorem~\ref{genericCentralized} is not tight for this problem. 
For instance, consider the simple canonical data shuffling setting, where the system has exactly $K$ files, all of equal size $F$ bits, and each node stores exactly one of these files, i.e., the entirety of the contents of the $k^{\text{th}}$ node $C_k$ is the $k^{\text{th}}$ file. 
Here, $|C_k|=F$ for all $k \in [K]$, and $C_i \cap C_j = \phi$ for all $i \neq j$.
Assume that the shuffling problem is to move the file $C_{k+1}$ to node $k$, i.e., $D_{k}=C_{k+1}$, where we consider the index $K+1$ to be equal to $1$. 
This is a worst-case demand for data shuffling incurring the largest possible communication cost. For this set of demands, we have $c_{k}^{\{k+1\}}=F$ for all $k \in [K]$, and $c_{k}^{Q} = 0$ for all other choices of $k,Q$. In particular, $c_{k+1}^{\{k\}}=0 \neq c_k^{\{k+1\}}$. Clearly, the condition in Theorem~\ref{thm:tightness_centralized} does not hold for $S=\{k,k+1\}$. Hence, our lower bound is strictly less than $L^*$  for this data shuffling problem, and therefore is not tight.

\section{Conclusion}
\textcolor{black}{We have presented an information theoretic converse result for a generic data exchange problem, where the terminals contain some data in their local storage and want other data available at the local storage of other nodes. As a number of recently studied multi-terminal communication problems fall under this setting, we have used our general converse to obtain converses in many such settings, thus recovering many existing results and presenting some new results as well. Using a connection with index coding, we also presented some ideas on why and when our data exchange based converse can be loose in the index coding setting. It would be quite interesting to see if our converse result can be tightened further while still retaining a closed form expression,  so as to cover all known bounds for any existing setting that can be modelled in the data exchange framework. A lower bound for the communication load in a generic data exchange setting in the presence of coded storage bits would also be a prospective direction for future research in this area.}
% % % % % % % % % % % 
%\appendixtitles{yes} % Leave argument "no" if all appendix headings stay EMPTY (then no dot is printed after "Appendix A"). If the appendix sections contain a heading then change the argument to "yes".
\appendices
%\appendices
\section{Proof of Theorem \ref{genericlowerboundTheorem}}
\label{appendix:proof_main_thm}
\textcolor{black}{We assume that all the bits in the collection $B$ as in Definition \ref{dataexchangeproblemdefinition} are i.i.d uniformly distributed on $\{0,1\}.$} For a given communication scheme for the given data exchange problem, let $X_i\triangleq \phi_i(C_i)$ represent the codeword transmitted by node $i$. For a subset $S\subset [K],$ let $X_S\triangleq \cup_{i\in S}X_i$. Also, let $Y_S=\bigcup_{i\in S} (D_i\cup C_i) $. We first prove the following claim.
\begin{claim}
\label{Claim}
For any $S\subset [K]$, 
%%%
\begin{align}
\label{eqn211}
    H(X_S|Y_{\overline S})\geq 
    \sum_{P\subset S}\sum_{Q\subset S\setminus P}\frac{|P|}{|P|+|Q|-1} a_{P}^{Q},
\end{align}
%%%%
where $\overline{S}=[K]\setminus S$.
\end{claim}
%%% 

Applying $S=[K]$ to the above claim then gives Theorem \ref{genericlowerboundTheorem}, as $L^*\geq H(X_{[K]}).$ 

Now we prove Claim \ref{Claim}. For this, we use induction on $|S|.$ 

We take the base case to be $|S|=2$, as for $|S|=1$ the problem of data exchange is not well defined. Let $S=\{1,2\}$ without loss of generality. Then  LHS of (\ref{eqn211}) gives 
\begin{align}
\nonumber
H(X_1,X_2|Y_{\overline S})&\geq H(X_1|Y_{\overline S})+H(X_2|Y_{\overline S},X_1)\\ 
\label{eqn218}
&\geq H(X_1|Y_{\overline S},C_2)+H(X_2|Y_{\overline S},C_1)\\
\label{eqn219}
&= H(X_1,D_2|Y_{\overline S},C_2)+H(X_2,D_1|Y_{\overline S},C_1)\\
&\geq H(D_2|Y_{{\overline S}},C_2)+H(D_1|Y_{{\overline S}},C_1)\\
& \geq a_{\{2\}}^{\{1\}}+a_{\{1\}}^{\{2\}},
\end{align}
where (\ref{eqn218}) follows as conditioning reduces entropy and $H(X_1|C_1)=0$, (\ref{eqn219}) is true as $H(D_2|Y_{\overline S},C_2,X_1)=0$ and $H(D_1|Y_{\overline S},C_1,X_2)=0$. This proves the base case. 
%H(X_S|Y_{\overline S})\geq  a_{\{1\}}^{\{2\}}+a_{\{2\}}^{\{1\}}$$  which is true as each node in $\{1,2\}$ has to necessarily transmit those stored bits demanded by the other node.

We now assume that the statement is true for $|S|=t-1,$ and prove that it holds for $|S|=t$. We have the LHS of (\ref{eqn211}) satisfying the following relationships for $|S|=t$. 
%%%
\begin{align}
    \nonumber
    H(X_S|Y_{\overline S})&=\frac{1}{t}\sum_{k\in S}\left(H(X_{S\setminus k}|Y_{\overline S},X_k)+H(X_k|Y_{\overline S})\right)\\
     \label{eqn212}
     &\geq\frac{1}{t}\left(\sum_{k\in S}H(X_{S\setminus k}|Y_{\overline S},C_k)\right)+\frac{1}{t}H(X_S|Y_{\overline S})\\
     \nonumber
     &\geq\frac{1}{t-1}\sum_{k\in S}H(X_{S\setminus k}|Y_{\overline S},C_k)\\
     \label{eqn213}
     &=\frac{1}{t-1}\sum_{k\in S}H(X_{S\setminus k},D_k|Y_{\overline S},C_k)\\
     \label{eqn214}
      & \textstyle =\frac{1}{t-1}\sum_{k\in S}\left(\textstyle H(D_k|C_k,Y_{\overline S})+H(X_{S\setminus k}|Y_{\overline{S\setminus k}})\right),
\end{align}
where (\ref{eqn212}) follows because $H(X_k|C_k)=0$. In (\ref{eqn213}), we introduce $D_k$ freely, because 
\begin{align*}H(D_k|X_{S\setminus k},Y_{\overline S},C_k)&\leq H(D_k|X_{S\setminus k},C_{\overline S},C_k)\\
&\leq H(D_k|X_{S\setminus k},X_{\overline S},C_k)\\
&=0,
\end{align*}
where the last two statements follow because $H(X_{\overline S}|C_{\overline S})=0$ and from the decoding condition, respectively.
%%%%
We now interpret the two terms of (\ref{eqn214}). For the first term, we have,
\begin{align}
\nonumber
    \sum_{k\in S}H(D_k|C_k,Y_{\overline S}) = &\sum_{k\in S} \sum_{P'\subset S\setminus k}\sum_{Q\subset S\setminus (P'\cup k)}a_{k\cup P'}^{Q}\\
    \label{eqn215}
    &=\sum_{P\subset S}\sum_{Q\subset S\setminus P}|P|a_{P}^{Q},
\end{align}
where the last statement follows by noting that for a fixed choice of $P,Q$ we have $|P|$ choices for $(k,P')$ such that $P'\cup k=P$. 

Now, using the induction hypothesis for the last term of (\ref{eqn214}),
\begin{align}
\nonumber
   \sum_{k\in S}&H(X_{S\setminus k}|Y_{\overline{S\setminus k}})\\
   &\geq \sum_{k\in S} \sum_{P\subset S\setminus k}\sum_{Q\subset S\setminus (P\cup k)}\frac{|P|}{|P|+|Q|-1} a_{P}^{Q},\\
   \label{eqn216}
    &=\sum_{P\subset S}\sum_{Q\subset S\setminus P}(t-|P|-|Q|)\frac{|P|}{|P|+|Q|-1} a_{P}^{Q}
\end{align}
where the above follows by noting that for a fixed choice of disjoint subsets $P,Q$ of $S$, we have $|S|-|P|-|Q|$ choices for $k$ such that $P\subset S\setminus k$ and $Q\subset S\setminus (P\cup k).$

Using (\ref{eqn215}) and (\ref{eqn216}) we have
\begin{align*}
\text{RHS~}&\text{of (\ref{eqn214})}\\
&\geq \frac{1}{t-1} \sum_{P\subset S}\sum_{Q\subset S\setminus P}|P|\left(1+\frac{t-|P|-|Q|}{|P|+|Q|-1}\right)a_P^Q\\
&=\sum_{P\subset S}\sum_{Q\subset S\setminus P}\frac{|P|}{|P|+|Q|-1}a_P^Q,
\end{align*}
thus proving Claim \ref{Claim}, which also concludes the proof of the theorem. 
%%%%
\section{Proof of Proposition~\ref{prop:centralized:1}} \label{app:prop:centralized:1}

We will continue to use the notations used in the proof of Theorem~\ref{genericlowerboundTheorem}. 
Let $k \in [K]$, and $S=\{\gamma_k,\dots,\gamma_K\}$. Note that ${\overline S}=\{\gamma_1,\dots,\gamma_{k-1}\}$.
We will prove by induction on $|S|=K-k+1$ that
\begin{equation} \label{eq:induction_centralized}
H(X_0 | Y_{\overline S}) \geq \sum_{i=k}^{K} \, \sum_{\substack{P \subset \{\gamma_i,\dots,\gamma_K\} \\ \gamma_i \in P}} \, \sum_{Q \subset \{\gamma_{i+1},\dots,\gamma_K\}} c_P^Q. 
\end{equation} 
Then the result claimed in the proposition follows by using $S=\{\gamma_1,\dots,\gamma_K\} = [K]$, i.e., $k=1$.
When $|S|=1$, i.e., $k=K$ and $S= \{\gamma_K\}$, clearly~\eqref{eq:induction_centralized} is true, since 
\begin{equation*}
H(X_0|Y_{ \{\gamma_1,\dots,\gamma_{K-1}\} }) \geq a_{\{\gamma_K\}}^{\{0\}} = c_S^{\phi}.
\end{equation*} 
Now consider $S=\{\gamma_k,\dots,\gamma_K\}$. The induction hypothesis is
\begin{equation} \label{eq:centralized_proof:1}
H(X_0 | Y_{\overline{S \setminus \gamma_k}}) \geq \sum_{i=k+1}^{K} \, \sum_{\substack{P \subset \{\gamma_i,\dots,\gamma_K\} \\ \gamma_i \in P}} \, \sum_{Q \subset \{\gamma_{i+1},\dots,\gamma_K\}} c_P^Q. 
\end{equation}  
Using the fact $H(D_{\gamma_k} | Y_{\overline S}, C_{\gamma_k}, X_0) \leq H(D_{\gamma_k} | C_{\gamma_k}, X_0)=0$, we have
\begin{align}
H(X_0|Y_{\overline S}) &\geq H(X_0 | Y_{\overline S}, C_{\gamma_k}) \nonumber \\
&= H(X_0 | Y_{\overline S}, C_{\gamma_k}) + H(D_{\gamma_k} | Y_{\overline S}, C_{\gamma_k}, X_0) \nonumber \\
&= H(X_0, D_{\gamma_k} | Y_{\overline S}, C_{\gamma_k}) \nonumber \\
&= H(D_{\gamma_k} | Y_{\overline S}, C_{\gamma_k}) + H(X_0|C_{\gamma_k}, D_{\gamma_k},Y_{\overline S}) \nonumber \\
&= \sum_{\substack{P \subset \{\gamma_k,\dots,\gamma_K\} \\ \gamma_k \in P}} \, \sum_{Q \subset \{\gamma_{k+1},\dots,\gamma_K\}} \!\!\!\!\!\!\!\! c_P^Q + H(X_0|Y_{\overline{S \setminus \gamma_k}}) \label{eq:centralized_proof:2}
\end{align} 
We observe that~\eqref{eq:induction_centralized} follows from~\eqref{eq:centralized_proof:1} and~\eqref{eq:centralized_proof:2}.

\section{Proof of Theorem~\ref{thm:alpha_bound}} \label{app:thm:alpha_bound}
We prove this theorem by showing that $\alpha$ is both upper and lower bounded by the right hand side of~\eqref{eq:thm:alpha_bound}.

\emph{Upper Bound:} 
Assume that $\mathcal{H}$ is a largest generalized independent set. We will now identify a permutation $\pi=(\pi_1,\dots,\pi_K)$ corresponding to $\mathcal{H}$.
Let $\mathcal{I}_1=\mathcal{H}$, and observe that since $\mathcal{I}_1$ is itself a subset of $\mathcal{H}$, it must contain an information bit, say $w_{P,m}^Q$ that is demanded by a node, say $\pi_1$, and none of the bits in $\mathcal{I}_1 \setminus \{w_{P,m}^Q\}$ is available as side information at $\pi_1$.
For $k=2,\dots,K$, we sequentially identify $\pi_k$ as follows. We first define
\begin{equation*}
\mathcal{I}_k = \mathcal{H} ~~\setminus~~ \bigcup_{i < k} \bigcup_{P: \pi_i \in P} \bigcup_{Q \subset [K] \setminus P} \left\{w_{P,m}^Q~:~m=1,\dots,c_P^Q \right\},
\end{equation*}
which is $\mathcal{H}$ minus all the bits demanded by any of $\pi_1,\dots,\pi_{k-1}$. 
Thus, any bit in $\mathcal{I}_k$ is demanded by one or more of the nodes in $[K] \setminus \{\pi_1,\dots,\pi_{k-1}\}$.
Since $\mathcal{I}_k \subset \mathcal{H}$, it contains an information bit such that this bit is demanded by a node, say $\pi_k \in [K] \setminus \{\pi_1,\dots,\pi_{k-1}\}$, and the rest of $\mathcal{I}_k$ is not available as side information at $\pi_k$.

Observe that $\mathcal{H}=\mathcal{I}_1 \supset \mathcal{I}_2 \supset \cdots \supset \mathcal{I}_K$, and $\mathcal{I}_k \setminus \mathcal{I}_{k+1}$ is the set of bits in $\mathcal{H}$ that are demanded by $\pi_k$ but not by any of the nodes in $\pi_1,\dots,\pi_{k-1}$. Thus, 
\begin{equation*}
\mathcal{I}_1 \setminus \mathcal{I}_2,~\mathcal{I}_2 \setminus \mathcal{I}_3,\dots,~\mathcal{I}_{K-1} \setminus \mathcal{I}_K,~\mathcal{I}_K 
\end{equation*} 
form a partition of $\mathcal{H}$.
Here we have abused the notation to denote $\mathcal{I}_K$ by $\mathcal{I}_K \setminus \mathcal{I}_{K+1}$.
We also observe that for any choice of $k'$ none of the bits of $\mathcal{I}_{k'}$ is available as side information at $\pi_{k'}$. If $k > k'$, since $\mathcal{I}_{k} \subset \mathcal{I}_{k'}$, we deduce that none of the bits in $\mathcal{I}_{k}$ is available as side information at $\pi_{k'}$. Thus, we conclude that each bit in $\mathcal{I}_k \setminus \mathcal{I}_{k+1}$ is demanded by $\pi_k$ and is neither demanded by and nor available as side information at any of $\pi_1,\dots,\pi_{k-1}$.
Hence, $|\mathcal{I}_k \setminus \mathcal{I}_{k+1}|$ is upper bounded by the number of bits exclusively demanded by $\pi_k$ and possibly some subset of $\{\pi_{k+1},\dots,\pi_K\}$ and which are also exclusively available at some subset of $\{\pi_{k+1},\dots,\pi_K\}$, i.e.,
\begin{equation*}
|\mathcal{I}_k \setminus \mathcal{I}_{k+1}| \leq \sum_{\substack{P \subset \{\pi_k,\dots,\pi_K\} \\ \pi_k \in P}} \, \sum_{Q \subset \{\pi_{k+1},\dots,\pi_K\}} c_P^Q.
\end{equation*} 
This provides us the following upper bound,
\begin{align*}
\alpha &= |\mathcal{H}| = \sum_{k=1}^{K} |\mathcal{I}_k \setminus \mathcal{I}_{k+1}| \\
&\leq \sum_{k=1}^{K} \sum_{\substack{P \subset \{\pi_k,\dots,\pi_K\} \\ \pi_k \in P}} \, \sum_{Q \subset \{\pi_{k+1},\dots,\pi_K\}} c_P^Q \\
&\leq \max_{\gamma} ~~ \sum_{k=1}^{K} \sum_{\substack{P \subset \{\gamma_k,\dots,\gamma_K\} \\ \gamma_k \in P}} \, \sum_{Q \subset \{\gamma_{k+1},\dots,\gamma_K\}} c_P^Q,
\end{align*} 
where the maximization is over all permutations $\gamma$ of $[K]$.

\emph{Lower bound:} 
We derive the lower bound by showing that, for any permutation $\gamma$, the set
\begin{equation*}
\mathcal{H} = \bigcup_{k=1}^{K} \bigcup_{\substack{P \subset \{\gamma_k,\dots,\gamma_K\} \\ \gamma_k \in P}} \bigcup_{\substack{Q \subset \\ \{\gamma_{k+1},\dots,\gamma_K\}}} \left\{w_{P,m}^Q~:~m=1,\dots,c_P^Q \right\}
\end{equation*} 
is a generalized independent set. Then,
\begin{equation*}
\alpha \geq \max_{\gamma} |\mathcal{H}| = \max_{\gamma} ~~ \sum_{k=1}^{K} \sum_{\substack{P \subset \{\gamma_k,\dots,\gamma_K\} \\ \gamma_k \in P}} \, \sum_{Q \subset \{\gamma_{k+1},\dots,\gamma_K\}} c_P^Q.
\end{equation*} 

To show that $\mathcal{H}$ is a generalized independent set, consider any subset $\mathcal{I} \subset \mathcal{H}$. Let $k$ be the smallest integer such that $\mathcal{I}$ contains an information bit $w_{P,m}^Q$ with $\gamma_k \in P$, i.e., $k$ is the smallest integer such that $\gamma_k$ demands some information bit in $\mathcal{I}$. Hence, any other bit $w_{P',m'}^{Q'}$ in $\mathcal{I}$ must satisfy
\begin{equation*}
P' \subset \{\gamma_{k'},\dots,\gamma_K\}, Q' \subset \{\gamma_{k'+1},\dots,\gamma_K\} \text{ for some } k' \geq k.
\end{equation*} 
Clearly, this bit is not available as side information at $\gamma_k$.
Hence $\mathcal{H}$ is a generalized independent set.

\section{Proof of Theorem~\ref{thm:tightness_centralized}} \label{app:thm:tightness_centralized}

For unicast problems $c_P^Q > 0$ only if $|P|=1$. We abuse the notation mildly and use $c_k^Q$ to denote $c_{\{k\}}^Q$.

\emph{The Necessity Part:} 
The lower bound of Proposition~\ref{prop:centralized:1} for unicast problems is $\sum_{i=1}^{K} \sum_{Q \subset \{\gamma_{i+1},\dots,\gamma_K\}} c_{\gamma_i}^Q$. For brevity, we will denote this sum as $f(\gamma)$. 
For Theorem~\ref{genericCentralized} to be tight it is necessary that the bound of this theorem be equal to $\alpha$, i.e., the value of $f$ be the same for all permutations $\gamma$. 

We first show that $c_i^{\{j\}}=c_j^{\{i\}}$ for any $i \neq j$. Consider two permutations $\gamma$ and $\pi$ which differ only in the two coordinates $K-1,K$, given as, $\gamma_{K-1}=i,\gamma_K=j$ and $\pi_{K-1}=j,\pi_K=i$. Then
\begin{align*}
0&=f(\gamma) - f(\pi) = c_i^{\{j\}} + c_i^{\phi} + c_j^{\phi} - c_j^{\{i\}} - c_j^{\phi} - c_i^{\phi} \\
&=c_i^{\{j\}} - c_j^{\{i\}}.
\end{align*} 
This proves the result for $|S|=2$.
Next, we will assume that the necessity part is true for any $S \subset [K]$ of size less than or equal to $t$, and use induction to prove the result for $|S|=t+1$.

Given any $(t+1)$-set $S \subset [K]$ and any $k,k' \in S$ we will now show that $c_k^{S \setminus k} = c_{k'}^{S \setminus k'}$.  
Consider two permutations $\gamma,\pi$ that differ only in the two coordinates $K-t,K-t+1$, and
\begin{align*}
&~~\gamma_{K-t},\dots,\gamma_{K},\pi_{K-t},\dots,\pi_{K} \in S, \\
&\gamma_{K-t} = k, \gamma_{K-t+1}=k', ~\text{ and }~ \pi_{K-t}=k', \pi_{K-t+1}=k.
\end{align*} 
We observe that $S=\{\gamma_{K-t},\hdots,\gamma_K\}=\{\pi_{K-t},\hdots,\pi_{K}\},$ and 
\begin{align}
0 &= f(\gamma) - f(\pi)  \nonumber  \\
&=\sum_{Q \subset \{\gamma_{K-t+1,\dots,\gamma_K}\}} \!\!\!\!\!\!\! c_k^Q ~~~+ \!\!\!\! \sum_{Q \subset \{\gamma_{K-t+2,\dots,\gamma_K}\}} \!\!\!\!\!\!\! c_{k'}^Q - \!\!\!\!\!\! \sum_{Q \subset \{\pi_{K-t+1,\dots,\pi_K}\}} \!\!\!\!\!\!\! c_{k'}^Q ~~~- \!\!\!\! \sum_{Q \subset \{\pi_{K-t+2,\dots,\pi_K}\}} \!\!\!\!\!\!\! c_{k}^Q \nonumber \\
&= \sum_{Q \subset S \setminus k} c_k^Q + \sum_{Q \subset S \setminus \{k,k'\}}\!\!\! c_{k'}^Q - \sum_{Q \subset S \setminus k'} c_{k'}^Q - \sum_{Q \subset S \setminus \{k,k'\}} \!\!\! c_k^Q. \label{eq:tightness_centralized:1}
\end{align} 

We now argue that except for the two terms $c_k^{S \setminus k}$ and $-c_{k'}^{S \setminus k'}$ all other terms in~\eqref{eq:tightness_centralized:1} cancel out. Consider any term in the first summation of~\eqref{eq:tightness_centralized:1} with $|Q| \leq t-1$. If $k' \in Q$, then by the induction hypothesis, the term $-c_{k'}^{Q \cup k \setminus k'}$ present in the third summation will cancel $c_k^Q$. If $k' \notin Q$, then $k,k' \notin Q$, and the term $-c_k^Q$ in the fourth summation will cancel $c_k^Q$.
Similarly, every term $c_{k'}^Q$ in the second summation will cancel the corresponding term $-c_{k'}^Q$ in the third summation.
It is straightforward to observe that these correspondences between the positive and negative terms are unique, and hence, we are left with
% \begin{equation*}
$0 = c_k^{S \setminus k} - c_{k'}^{S \setminus k'}$.
% \end{equation*}

\emph{The Sufficiency Part:}
The lower bound in Theorem~\ref{genericCentralized} is 
\begin{align*}
L^* \geq \sum_{k=1}^{K} \sum_{Q \subset [K] \setminus k} \frac{1}{1 + |Q|} c_k^Q 
= \sum_{k=1}^{K} c_k^{\phi} + \sum_{\substack{S \subset [K] \\ |S| \geq 2}} \sum_{k \in S} \frac{c_k^{S \setminus k}}{|S|}.
\end{align*} 
This lower bound can be met by a scheme that uses a combination of uncoded transmission and clique covering.
All the bits that are not available at any of the $K$ clients are transmitted uncoded incurring the cost $\sum_{k=1}^{K} c_k^{\phi}$. 
For every $S \subset [K]$ with $|S| \geq 2$, the encoder constructs $|S|$ vectors, one corresponding to each $k \in S$ and broadcasts the XOR of these vectors to the clients. The vector for $k \in S$ consists of the $c_k^{S \setminus k}$ bits demanded by node $k$ and available at nodes $S \setminus k$. All these $|S|$ vectors have the same length $c_k^{S \setminus k}$. These coded transmissions incur an additional cost $\sum_{\substack{S \subset [K] \\ |S| \geq 2}} \sum_{k \in S} \frac{c_k^{S \setminus k}}{|S|}$, thereby achieving the lower bound. This is the well known clique-covering index coding scheme (see \cite{BirkandKolISCOD,BipartiteIndexCoding}) and these transmissions allow the clients to decode their demands.
%%%%%
\bibliographystyle{IEEEtran}
\bibliography{biblio.bib}
\end{document}